\documentclass[letterpaper, 10 pt, conference]{ieeeconf}  

\IEEEoverridecommandlockouts                              
\title{\LARGE \bf Planar Cooperative Extremum Seeking with Guaranteed Convergence Using A Three-Robot Formation}

\author{Anna Skobeleva, Baris Fidan, Valeri Ugrinovskii, Ian R. Petersen
\thanks{Anna Skobeleva and Valeri Ugrinovskii are with School of Engineering and Information Technology, University of New South Wales, Canberra, Australia. Baris Fidan is with the Mechanical and Mechatronics Engineering Department, University of Waterloo, ON, Canada. Ian Petersen is with College of Engineering and Computer Science, Australian National University, Canberra, Australia. 
{\tt\small anna.skobeleva@student.adfa.edu.au}}
\thanks{This research was supported by the Australian Research Council under grant DP160101121}
}

\usepackage{graphicx}
\usepackage{eqnarray, amsmath, amssymb}
\usepackage{enumerate}
\usepackage{cite}
\usepackage{color}

\newtheorem{theorem}{Theorem}
\newtheorem{lemma}{Lemma}
\newtheorem{corollary}{Corollary}
\newtheorem{prop}{Proposition}

\newtheorem{assumption}{Assumption}

\newtheorem{problem}{Problem}

\begin{document}


\maketitle
\thispagestyle{plain}
\pagestyle{plain}

\begin{abstract}

In this paper, a combined formation acquisition and cooperative extremum seeking control scheme is proposed for a team of three robots moving on a plane. The extremum seeking task is to find the maximizer of an unknown two-dimensional function on the plane. The function represents the signal strength field due to a source located at maximizer, and is assumed to be locally concave around  maximizer and monotonically decreasing in distance to the source location. Taylor expansions of the field function at the location of a particular lead robot and the maximizer are used together with a gradient estimator based on signal strength measurements of the robots to design and analyze the proposed control scheme. The proposed scheme is proven to exponentially and simultaneously (i) acquire the specified geometric formation and (ii) drive the lead robot to a specified neighborhood disk around maximizer, whose radius depends on the specified desired formation size as well as the norm bounds of the Hessian of the field function. The performance of the proposed control scheme is evaluated using a set of simulation experiments.
\end{abstract}

\section{Introduction}

The problem of signal source localization using mobile sensory agents has been studied from various perspectives and following various technical approaches \cite{Patwari05,Cao06,Path05,Roum02,Dandach08_SCL,Fidan08_TrSP,cochran2009nonholonomic,moore2010source}.
In source localization, the field of interest is unknown and thus the main challenge is to estimate the gradient of the field and then choose an appropriate algorithm that would drive the agent or the formation to a point where the gradient is zero. 

The classical extremum seeking control approach studied in \cite{ariyur2003real}, \cite{zhang2007extremum}, \cite{zhang2007source}, \cite{cochran2009nonholonomic}, \cite{lin2017stochastic} employs a zero-mean dither signal to extract gradient information from the field measurements. A dither signal is introduced as a part the velocity control and can be of a sinusoidal \cite{zhang2007source},\cite{cochran2009nonholonomic} or stochastic \cite{lin2017stochastic} shape. While the above technique does not require any position information, the motion pattern is  inefficient and agents are unable to come to a complete stop at the source location, but rather they continue moving in its vicinity. 

Examples of dither free extremum seeking control techniques for a single agent can be found in \cite{matveev2011navigation}, \cite{fu2008sliding}, \cite{mayhew2007robust}. The authors of \cite{matveev2011navigation} and \cite{fu2008sliding} propose to estimate the field's gradient with a difference of two time distinct measurements taken by the agent and a sliding mode controller for the agent's angular velocity. While sliding mode control is robust to measurement noise and computationally efficient, the chattering effect is a well-known drawback of this approach.

Using  a group of agents instead of a single robot allows for better gradient estimation using a combination of the measurements across the platforms and eliminates the need for auxiliary movements making the search more time and energy efficient. In \cite{ogren2004cooperative}, the gradient is estimated from distributed field measurements using least squares and is further refined by applying a Kalman filter to the history data. The virtual leader is then moving in the direction of the steepest descent/ascent of the gradient. In their subsequent work \cite{zhang2010cooperative}, the authors propose Kalman filter schemes to estimate both the gradient and the Hessian of the unknown field. A significant drawback of this approach is its high computational complexity and sensitivity to communication delays and faults.

A source seeking approach utilizing a circular formation of unicycle-like agents is considered in the series of papers \cite{moore2010source}, \cite{brinon2011collaborative}, \cite{brinon2013consensus}, \cite{brinon2016distributed}. These papers exploit the circular shape of the formation with agents being uniformly distributed around the circle to approximate the field's gradient at the formation centre as an average of the weighted measurements taken by the agents. The reference trajectory for the formation centre is calculated  by integrating the estimated gradient value. An  additional consensus algorithm is used to agree on the estimated gradient direction. In \cite{brinon2013consensus}, the formation center trajectory is generated with a gradient-ascent algorithm, and in \cite{brinon2016distributed}, the authors provide simulation results for scenarios with noisy field measurements, multiple maximum points, and time-varying fields. In the simulation results of \cite{brinon2016distributed}, the formation converges to a neighbourhood of the maximum point or one of the local maximum points, for such scenarios. 

The signal field extremum seeking task of our current paper is the same as that of the aforementioned papers, and we also use a formation control basic cooperative approach. However, different from earlier works, (i) we focus on use of a seed formation with minimal number (three) of robot agents needed for providing static estimates of the field gradient, (ii) we consider formation acquisition and extremum seeking as simultaneous control goals without assuming satisfaction of the desired geometric formation initially, (iii) we formally establish guaranteed convergence results to a certain neighborhood of the extremum point.

The rest of the paper is organized as follows: The simultaneous formation acquisition and extremum seeking control problem is defined in Section II. Section III provides background on the Taylor expansion of functions of vectors, the proposed distributed control design and convergence analysis. Section IV provides the results of simulation tests. The paper is concluded with the final remarks in Section V.

\section{Problem Definition}

The main task we study in this paper is to have a team of three robot agents $A_0$, $A_1$, $A_2$ to search for and move towards a signal source located at an unknown position $x^* \in {\Re ^2}$. Leaving the detailed motion dynamics, low level dynamic control design, and implementation issues to future studies, we consider the following velocity integrator kinematics of the robot agents in this paper:
\begin{equation}
\label{eq:xiDot}
\dot{x}_i(t)=v_i(t),~~~i\in\{0,1,2\},
\end{equation}
where $x_i(t)~=~[x_{ix}(t),x_{iy}(t)]^T\in{\cal P}$ and $v_i(t)~=~[v_{ix}(t),v_{iy}(t)]^T\in\Re^2$ denote, respectively, the position and velocity of agent $A_i$ at time instant $t$.

Simultaneously with the above main signal source seeking task, the three robot agents are desired to acquire and maintain a pre-defined geometric formation defined in terms of the desired values $r_i^*$ of their relative positions
\begin{equation}
\label{eq:ri}
r_i(t)=[r_{ix}(t),r_{iy}(t)]^T=x_i(t)-x_0(t),~~~i\in\{1,2\}.
\end{equation}

The signal strength at any point $x~=~[x_x,x_y]^T\in{\cal P}$  due to the signal source at $x^*$ is denoted by $f(x)$, where $f(\cdot):\Re^2\mapsto [0, \infty)$ is an unknown function which satisfies the following assumptions.

\begin{assumption} \label{as1}
\begin{enumerate}[(i)]
\item The function $f$, its gradient $\nabla f = \left [ \partial_x f,~\partial_y f \right ]^T = \left [ \frac{\partial f}{\partial x_x},~\frac{\partial f}{\partial x_y} \right ]^T$, and its Hessian
$$\nabla^2 f =\begin{bmatrix}
                        \partial_{xx}f & \partial_{xy}f \\
                        \partial_{xy}f & \partial_{xy}f \\
                      \end{bmatrix}
                      =\begin{bmatrix}
                        \frac{\partial^2 f}{\partial x_x^2} & \frac{\partial^2 f}{\partial x_x\partial x_y} \\
                       \frac{\partial^2 f}{\partial x_x\partial x_y} & \frac{\partial^2 f}{\partial x_y^2} \\
                      \end{bmatrix}$$
are continuous, and the entries of $\nabla^2 f$ are continuously differentiable.
\item The function $f$ has a single maximum at a point $x^*$. 
\item There exists a scalar $M_H>0$ such that $\forall~x \in {\Re^2}$, $\|\nabla^2 f(x)\|\leq M_H$.
\item At $x^*$, $\nabla^2 f(x^*)$ is negative definite.
\item There exists a scalar $L_H>0$ such that $\forall~x_1,x_2\in {\Re^2}$,
$\|\nabla^2 f(x_1)-\nabla^2 f(x_2)\| \leq L_H \|x_1-x_2\|$.

\item There exists a scalar $G_H>0$ such that $\forall~x_1,x_2\in {\Re^2}$,
$\|\nabla^2 f(x_1)-\nabla^2 f(x_2)\| \leq G_H$.
\end{enumerate}
\end{assumption}

Next, we formulate the simultaneous formation acquisition and extremum seeking problem explained above.

\begin{problem}
Consider three robot agents $A_0$, $A_1$, $A_2$ with motion kinematics \eqref{eq:xiDot}, and a signal source located at an unknown position $x^* \in {\Re ^2}$. The signal strength distribution due to this source is represented by an unknown function $f(\cdot):\Re^2\mapsto [0, \infty)$ which satisfies \textit{Assumption} \ref{as1}. Assume that the agents can sense their positions $x_i(t)$ and the signal strengths $f_i(t)=f(x_i(t))$ ($i\in\{0,1,2\}$) and communicate among themselves. The problem is to design control laws to produce the velocities $v_i$ such that $x_0(t)$ converges to a bounded neighbourhood of the unique maximizer $x^*$ of $f(\cdot)$.
\end{problem}

\begin{figure}[h!]
\centering
\includegraphics[scale=0.4]{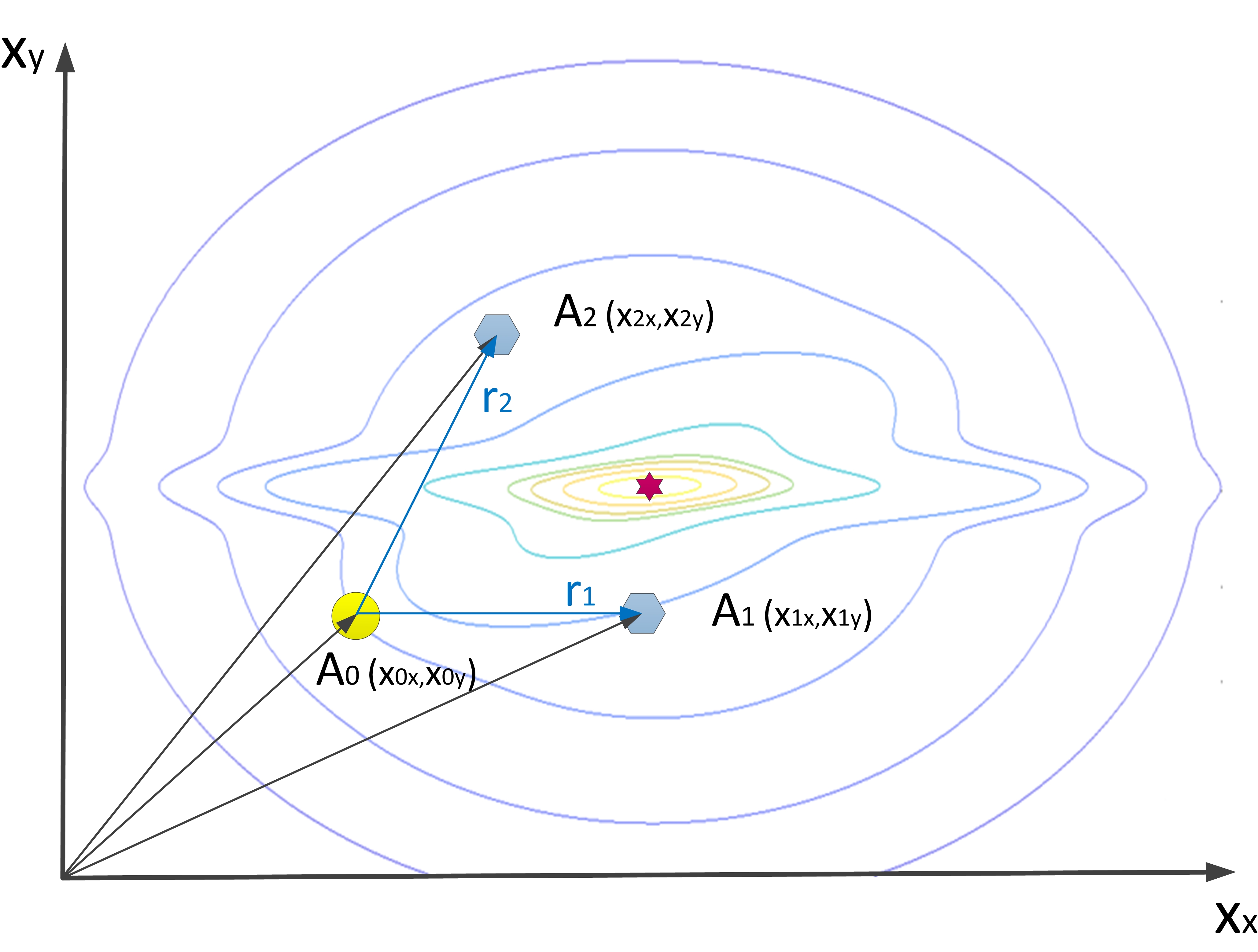}
\caption{\small{Illustration of Problem 1.}}\label{fig:2Dsetup}
\end{figure}

\section{Simultaneous Extremum Seeking and Formation Acquisition }
\subsection{Taylor Expansions of Functions of 2D Vectors}

The control design we present in the following subsections to solve Problem 1 utilizes Taylor expansions for the function $f(\cdot)$ and its gradient $\nabla f(\cdot)$. In this subsection, we summarize these expansion formulations. The following fact is a special case of Theorem 5.2 in \cite{Lang87} for $f(\cdot):\Re^2\mapsto [0, \infty)$:

\begin{prop} \label{prp:Taylor1}
Assume that $f$ is continuously differentiable up to order 2 within a given open set ${\cal X}\subset \Re^2$. Consider two vectors $\bar{x},x\in{\cal X}$ and the difference vector $h=x-\bar{x}$. There exists a number $0\leq \tau \leq 1$ such that
\begin{eqnarray*}
f(x)&=&f(\bar{x})+h^T\nabla f(\bar{x})+\frac{1}{2}h^T \nabla^2 f(\bar{x}+\tau h) h.
\end{eqnarray*}
\end{prop}

\textit{Proposition} \ref{prp:Taylor1} further has the following corollary:
\begin{corollary} \label{cor:Taylor1}
Assume that $f$ is continuously differentiable up to order 3 within a given open set ${\cal X}\subset \Re^2$. Consider two vectors $\bar{x},x\in{\cal X}$ and the difference vector $h=x-\bar{x}$. There exist numbers $0\leq \tau_1, \tau_2 \leq 1$ such that
\begin{eqnarray*}
\nabla f(x)&=&\nabla f(\bar{x})+\hat{\nabla}^2 f(\bar{x}+\tau_1 h,\bar{x}+\tau_2 h)h,
\end{eqnarray*}
where $\hat{\nabla}^2 f(\xi_1,\xi_2)$ for $\xi_1,\xi_2\in{\cal X}$ is defined by
\begin{eqnarray} \label{eq:nabla2hat}
\hat{\nabla}^2 f(\xi_1,\xi_2)&=&
\begin{bmatrix}
                        \partial_{xx}f(\xi_1) & \partial_{xy}f(\xi_1) \\
                        \partial_{xy}f(\xi_2) & \partial_{xy}f(\xi_2) \\
                      \end{bmatrix}.
\end{eqnarray}
\end{corollary}

\subsection{Control Algorithm}
Extending the gradient search and Taylor expansion based approach of \cite{skobeleva2017esc1d} to the two dimensional setting of this paper, we propose the following distributed control law for the agents $A_0$, $A_1$, $A_2$:

\begin{equation}
\begin{aligned}
v_0(t)&=K_0g(t),\\
v_1(t)&=-K_1(r_1(t)-r_1^*) + K_0 g(t),\\
v_2(t)&=-K_2(r_2(t)-r_2^*) + K_0 g(t),
\end{aligned}
\label{eq:vi}
\end{equation}
where $K_i=k_i I_2$ for some positive gain scalar gain $k_i$, $r_i(t)$ is the measurable relative position as defined in (\ref{eq:ri}) and $r_i^* $ is its desired value for the formation acquisition task in Problem 1, for $i\in\{1,2\}$, and
\begin{equation}
g(t)=
R^{-1}(t)
\begin{bmatrix}
f_1(t)-f_0(t) \\ f_2(t)-f_0(t)
\end{bmatrix},~~R(t)=\begin{bmatrix}
r_1^T(t)\\
r_2^T(t)
\end{bmatrix}
\label{eq:g}
\end{equation}
is an instantaneous approximation of the gradient $\nabla f(x_0(t))$, calculated using the measurements of $x_i(t),f_i(t)$, $i\in\{0,1,2\}$. Note that for the instantaneous approximation \eqref{eq:g} to be well defined the matrix $R(t)$ needs to be invertible. The invertibility of $R(t)$ is analyzed in the next subsection.

\subsection{Invertibility of $R(t)$}
In this subsection we establish some practical conditions for the formation matrix $R(t)$ in \eqref{eq:g} to be invertible. These conditions will ensure that the control algorithm (\ref{eq:vi}), (\ref{eq:g}) is implementable.
\begin{assumption} \label{as:formation}
Given initial positions of the agents, choose indexing $A_{0}$, $A_{1}$, $A_{2}$ of these agents and the orientation of the formation described by the matrix
\begin{equation}
R^*=\begin{bmatrix}
r_1^{*T}\\
r_2^{*T}
\end{bmatrix}
\label{eq:Rs}
\end{equation}
such that:
\begin{enumerate}[(i)]
\item $r_1^*$ and $r_{10}=r_1(0)$ have the same direction;
\item $x_2^*(0)=x_0(0)+r_2^*$ and $x_2(0)$ are on the same half-plane with respect to the $x_{0}(0)x_{1}(0)$.
\end{enumerate}
\end{assumption}

Figure \ref{fig:formConfig} illustrates a formation setting satisfying \textit{Assumption} \ref{as:formation}.

\begin{figure}[h!]
\centering
\includegraphics[scale=0.5]{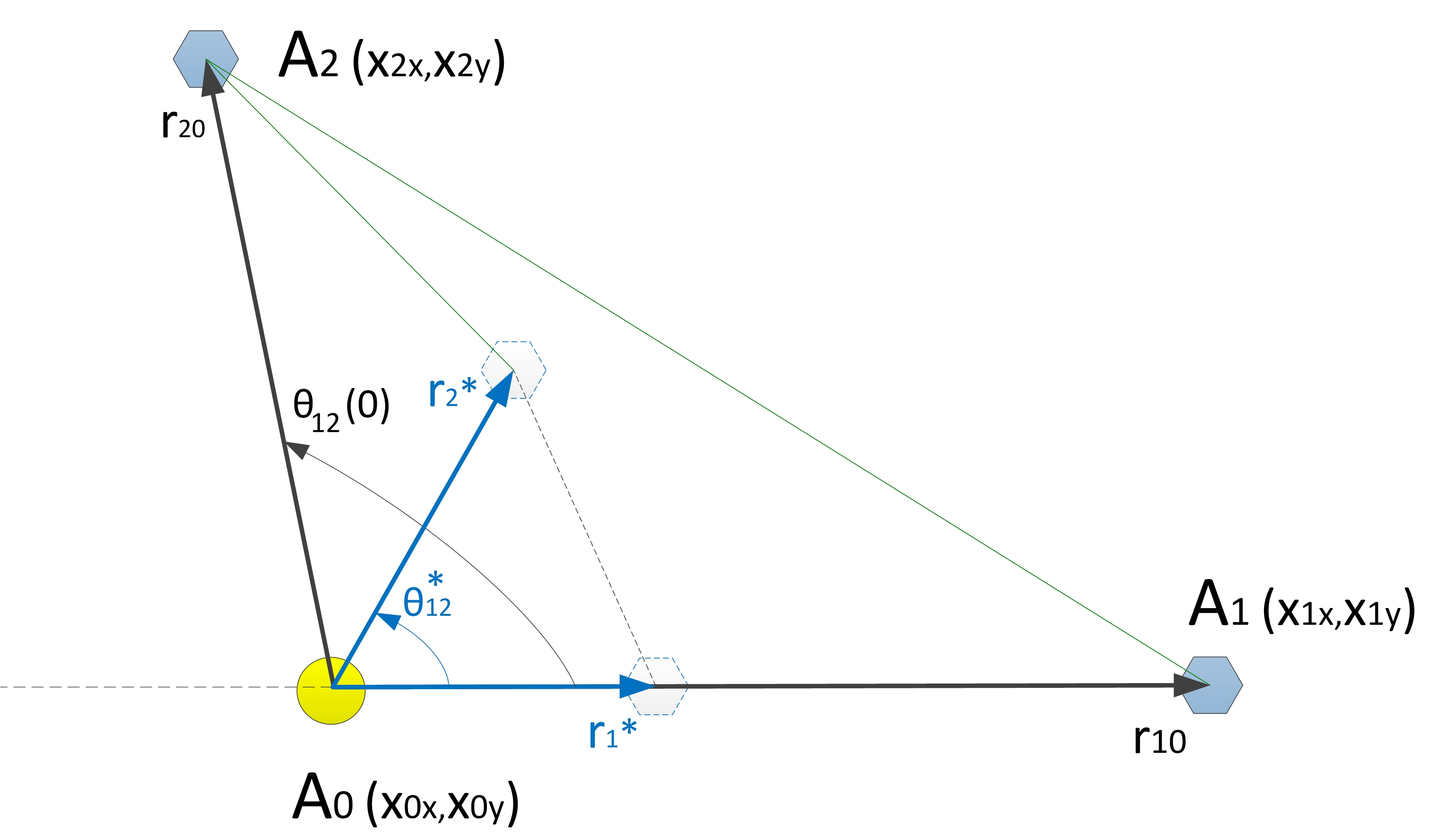}
\caption{\small{A target formation setting satisfying \textit{Assumption} \ref{as:formation}. }}\label{fig:formConfig}
\end{figure}

\begin{lemma} \label{lm:formation}
Denote the counterclockwise angles from $r_1^*$ to $r_2^*$ and from $r_1(t)$ to $r_2(t)$, respectively, by $\theta _{12}^*$ and $\theta _{12}(t)$. If \textit{Assumption} \ref{as:formation} is satisfied, for all $t\geq 0$, the sign of $\sin (\theta _{12}(t))$ is the same as the signs of
$$\rho^*=\sin (\theta _{12}^*),~\rho_0=\sin (\theta _{12}(0))$$ 
and the following inequality holds:
$$|\sin (\theta _{12}(t))| \geq \min (|\rho^*|, |\rho_0|).$$ 
\end{lemma}

\begin{proof}
If $\theta _{12}^* \in (0, \pi)$ is as illustrated in Figure \ref{fig:formConfig}, then by \textit{Assumption} \ref{as:formation} $(ii)$, $\theta _{12}(0)\in (0, \pi)$. Using the notation
$$\tilde{x}_i(t)=x_0(t)+r_i^*$$
and observing the triangle $x_0(t)\tilde{x}_1(t)x_2(t)$, that is formed by vectors $r_1^*$ and $r_2(t)$, we have:
\begin{enumerate}[1.]
\item If $\theta _{12}(0) >\theta _{12}^*$, then $\theta _{12}^* \leq \theta _{12}(t) \leq \theta _{12}(0)$, $\forall t$;
\item If $\theta _{12}(0) \leq \theta _{12}^*$, then $\theta _{12}(0) \leq \theta _{12}(t) \leq \theta _{12}^*$, $\forall t$.
\end{enumerate}

We can repeat the analysis above for the case when $\theta _{12}^* \in (\pi, 2 \pi)$. We conclude that $\theta _{12}(t)$ is always between $\theta _{12}(0)$ and $\theta _{12}^*$. Therefore,  $ \sin (\theta _{12}(t)) \geq \min (\sin (\theta _{12}^*),\sin (\theta _{12}(0))) \geq 0$  $\forall t$.\\
\end{proof}

\begin{corollary} \label{cor:formation}
For any $t \geq 0$,
\begin{equation*}
\left | det (R(t)) \right | \geq \min(|\rho _0|, |\rho ^*|) \|r_{1}(t)\| \|r_{2}(t)\|.
\end{equation*}
\end{corollary}

\begin{proof} 
Note that $det (R(0)) = \rho_0 \|r_{1}(0)\| \|r_{2}(0)\|$ and $det (R^*) = \rho^* \|r_{1}^*\| \|r_{2}^*\|$. Hence,
the result directly follows from Lemma \ref{lm:formation} and the fact that
\begin{equation*}
\left | det (R(t)) \right | = |r_1(t) \times r_2(t)|= \|r_{1}(t)\|\|r_{2}(t)\| |\sin \theta_{12}(t)|.
\end{equation*}
\end{proof}

We conclude that if the \textit{Assumption} \ref{as:formation} is satisfied, then $R(t)$ is non singular. Therefore, $R^{-1}(t)$ exists and (\ref{eq:g}) is implementable for all $t \geq 0$.

\subsection{Formation Acquisition and Extremum Seeking Convergence Analysis}
In this subsection we show that the formation control task is achieved exponentially fast. Then, we analyze stability and convergence properties of the extremum seeking dynamics.
 
Defining relative position errors with respect to the desired formation
\begin{equation}
\delta_i(t) = r_i(t)-r_i^*,~i\in\{1,2\},
\label{eq:deltai}
\end{equation}
from (\ref{eq:ri}), (\ref{eq:vi}), we have
\begin{equation}
\dot{\delta}_i(t) =\dot{x}_i(t)-\dot{x}_0(t)=-K_i \delta_i(t),~i\in\{1,2\},
\label{eq:deltaiDot}
\end{equation}
and hence
\begin{equation}
\delta_i(t) =\delta_{i0}e^{-k_it},\forall t\geq 0, ~i\in\{1,2\},
\label{eq:deltaioft}
\end{equation}
where $\delta_{i0}=\delta_{i}(0)$. \eqref{eq:deltaioft} further implies that
\begin{equation}
r_i(t) =r_i^*+\delta_{i0}e^{-k_it},\forall t\geq 0, ~i\in\{1,2\}.
\label{eq:rioft}
\end{equation}

Next, we analyze of the extremum seeking convergence properties.
Combining (\ref{eq:xiDot}), \eqref{eq:vi} and (\ref{eq:g}), the dynamics of the agent $A_0$ are given by
\begin{equation}
\dot{x}_0(t)=
K_0
R^{-1}(t)
\begin{bmatrix}
f_1(t)-f_0(t) \\ f_2(t)-f_0(t)
\end{bmatrix}.
\label{eq:x0Dot}
\end{equation}
To analyse the right hand side of equation (\ref{eq:x0Dot}), we use the Taylor expansions of $f$ and $\nabla f$ around $x_0$ and $x^*$, respectively. First, we obtain an expansion of $f_i=f(x_i)$ around $x_0$ with a reminder in the integral form \cite{marsden2003vector}:
\begin{equation}
f_i-f_0=r_i^T \nabla f (x_0)+\int _0 ^1 (1-\xi_i)r_i^T \nabla ^2f(x_0+\xi_i r_i)r_i d\xi_i,
\label{eq:fif0}
\end{equation}
for $i\in\{1,2\}$, where $\xi_i$ comes from parametrization of the segment $[x_0,x_i]$: $x_i=x_0+\xi_i r_i$.

Later in the analysis, we will consider a special case of \eqref{eq:fif0} corresponding to $\delta_i=0$. In this case, $r_i=r_i^*$, $x_i=x_0+r_i^*$ and $f_i=f_i^*=f(x_0+r_i^*)$, and \eqref{eq:fif0} becomes
\begin{equation}
f_i^*-f_0=r_i^{*T} \nabla f (x_0)+\int _0 ^1 (1-\xi_i)r_i^{T*} \nabla ^2f(x_0+\xi_i r_i^*)r_i^*d\xi_i,
\label{eq:fif0s}
\end{equation}
for $i\in\{1,2\}$.

(\ref{eq:x0Dot}) and (\ref{eq:fif0}) imply that
\begin{equation}
\begin{aligned}
\dot{x}_0=&
K_0\nabla f (x_0)\\
&+\frac{1}{2}K_0
R^{-1}
\begin{bmatrix}
\int _0 ^1 (1-\xi_1)r_1^T \nabla ^2f(x_0+\xi_1 r_1)r_1 d\xi_1 \\
\int _0 ^1 (1-\xi_2)r_2^T \nabla ^2f(x_0+\xi_2 r_2)r_2 d\xi_2
\end{bmatrix}.
\label{eq:x0Dot2}
\end{aligned}
\end{equation}

Next, defining the extremum seeking error
\begin{equation}
z(t) =x_0(t)-x^*,
\label{eq:z}
\end{equation}
we apply Corollary \ref{cor:Taylor1} to further expand the term $\nabla f (x_0)$ in \eqref{eq:x0Dot2} around the unknown extremum point $x^*$:

\begin{equation}
\begin{aligned}
\nabla f (x_0)&=\nabla f (x^*)+ \hat{\nabla} ^2 f(\xi_{01},\xi_{02})(x_0-x^*) \\
&=\nabla ^2 f(x^*)z+(\hat{\nabla} ^2 f(\xi_{01},\xi_{02})-\nabla ^2 f(x^*))z,
\end{aligned}
\label{eq:TaylorGradX0}
\end{equation}
where $\xi_{01},\xi_{02}$ are certain points on the segment $[x^*,x_0]$, noting that $\nabla f (x^*)=0$ at the extremum point $x^*$.

Combining (\ref{eq:x0Dot2}), (\ref{eq:z}), (\ref{eq:TaylorGradX0}), we obtain

\begin{equation}
\begin{aligned}
\dot{z}=&\dot{x}_0=
K_0 \nabla^2 f (x^*)z+B_{\beta}z\\
&+\frac{1}{2}K_0
R^{-1}
\begin{bmatrix}
\int _0 ^1 (1-\xi_1)r_1^T \nabla ^2f(x_0+\xi_1 r_1)r_1 d\xi_1 \\
\int _0 ^1 (1-\xi_2)r_2^T \nabla ^2f(x_0+\xi_2 r_2)r_2 d\xi_2
\end{bmatrix},
\end{aligned}
\label{eq:zDot}
\end{equation}
where
\begin{equation}
B_{\beta}=K_0\left ( \hat{\nabla} ^2 f(\xi_{01},\xi_{02})-\nabla ^2 f(x^*) \right ).
\label{eq:Bbeta}
\end{equation}

To combine the differential equations for the formation and extremum seeking errors in a single compact equation, we further define $\delta = [\delta_1^T,\delta_2^T]^T$ and the six dimensional stacked error vector
\begin{equation}\label{eq:chi}
 \chi=\begin{bmatrix}\delta^T & z^T \end{bmatrix}^T=\begin{bmatrix}\delta_1^T & \delta_2^T & z^T \end{bmatrix}^T,
\end{equation}
and consider an uncertain system:
\begin{equation}
\begin{aligned}
\dot{\chi}=&
\underbrace{\begin{bmatrix}
-K_1 &O_2 &O_2\\ O_2&-K_2 &O_2\\O_2 &O_2 & K_0\nabla^2 f (x^*)
\end{bmatrix}}_{A}\chi \\
&+\underbrace{\begin{bmatrix}
O_2\\O_2\\I_2
\end{bmatrix}}_{B} (\phi_1 + \phi_2(t))\\&=A\chi +B(\phi_1+\phi_2^*+\phi_3(t)),
\label{eq:ChiDot}
\end{aligned}
\end{equation}

with three uncertainty inputs $\phi_1$, $\phi_2^*$ and $\phi_3(t)$ defined as:

\begin{equation}
\begin{aligned}
&\phi_1=B_{\beta}z=B_{\beta}C_1 \chi,~~ C_1=\begin{bmatrix}O_2&O_2&I_2\end{bmatrix};\\
&\phi_2(t)= \frac{1}{2}K_0
R^{-1}
\begin{bmatrix}
\int _0 ^1 (1-\xi_1)r_1^T \nabla ^2f(x_0+\xi_1 r_1)r_1 d\xi_1 \\
\int _0 ^1 (1-\xi_2)r_2^T \nabla ^2f(x_0+\xi_2 r_2)r_2 d\xi_2
\end{bmatrix};\\
&\phi_2^*=\frac{1}{2}K_0
R^{*-1}
\begin{bmatrix}
\int _0 ^1 (1-\xi_1)r_1^{T*} \nabla ^2f(x_0+\xi_1 r_1^*)r_1^* d\xi_1 \\
\int _0 ^1 (1-\xi_2)r_2^{T*} \nabla ^2f(x_0+\xi_2 r_2^*)r_2^* d\xi_2
\end{bmatrix},\\
& R^{*-1}= \begin{bmatrix}r_1^{*T}\\r_2^{*T} \end{bmatrix} ^{-1};\\
&\phi_3(t)=\phi_2(t)-\phi_2^*.
\label{eq:phiDef}
\end{aligned}
\end{equation}

\begin{lemma} \label{lm:phiBounds}
The uncertainty inputs $\phi_1(t), ~\phi_2^*$ in (\ref{eq:phiDef}) satisfy 
\begin{equation}
\begin{aligned}
&\|\phi_1(t)\| \leq  \alpha_1 \|C_1 \chi\|,~~\alpha_1=2G_Hk_0, \\
&\|\phi_2^*\| \leq  \alpha_2^* ,\\
&\alpha_2^*=\frac{1}{2}\Bigg\|K_0 R^{*-1}
\begin{bmatrix}
\int _0 ^1 (1-\xi_1)r_1^{T*} \nabla ^2f(x_0+\xi_1 r_1^*)r_1^* d\xi_1 \\
\int _0 ^1 (1-\xi_2)r_2^{T*} \nabla ^2f(x_0+\xi_2 r_2^*)r_2^* d\xi_2
\end{bmatrix}\Bigg\| \\
&\leq \frac{k_0M_H}{2}\| R^{*-1} \| ( \|r_1^*\|^2 + \| r_2^* \|^2 ),\\
&\forall t\geq 0.
\label{eq:constraints1}
\end{aligned}
\end{equation}
Further, the uncertainty input $\phi_3(t)$ is exponentially decaying to zero
\begin{equation}
\|\phi_3(t)\| =  \varepsilon e^{-\beta t},
\label{eq:constraints2}
\end{equation}
where $\varepsilon, ~ \beta >0$.
\end{lemma}

Proof of \textit{Lemma} \ref{lm:phiBounds} is provided in the Appendix.

Further analysis of the error dynamics in \eqref{eq:ChiDot} is carried out using the approach based on the robust stability theory, employing the Lyapunov function and S-procedure techniques.

Choose $V(\chi)=\chi^TP\chi$ as a candidate Lyapunov function. Let $\tau$, $\gamma_1,\gamma_2>0$ be three constants. Further consider: 
\begin{equation}
\begin{aligned}
&\dot{V} -\tau(\|\phi_1\|^2-\alpha_1 ^2\|C_1\chi\|^2) =\\
&2\chi^T P(A\chi+B\phi_1+B\phi_2^* +B\phi_3)-\tau(\|\phi_1\|^2-\alpha_1 ^2\chi ^TC_1^TC_1\chi)\\
&=\chi^T(PA+A^TP+\tau\alpha_1 ^2C_1^TC_1)\chi+2\chi^T PB\phi_1\\
&+2\chi^TPB\phi_2^* +2\chi^T PB\phi_3-\tau\|\phi_1\|^2.
\end{aligned}
\label{eq:Vdot}
\end{equation}

Using the fact that
\begin{equation*}
\begin{aligned}
2 \chi^TPB \phi_1=2 \frac{1}{\sqrt{\tau}}\chi^TPB \sqrt{\tau}\phi_1 \leq & \frac{1}{\tau}\chi ^TPBB^TP\chi\\ &+ \tau\|\phi_1\|^2, \\
2 \chi^TPB \phi_2^*=2 \frac{1}{\sqrt{\gamma_1}}\chi^TPB \sqrt{\gamma_1}\phi_2^* \leq& \frac{1}{\gamma_1}\chi ^T PBB^TP\chi\\ &+ \gamma_1\|\phi_2^*\|^2,\\
2 \chi^TPB \phi_3=2 \frac{1}{\sqrt{\gamma_2}}\chi^TPB \sqrt{\gamma_2}\phi_3 \leq& \frac{1}{\gamma_2}\chi ^T PBB^TP\chi \\ &+ \gamma_2\|\phi_3\|^2, 
\end{aligned}
\end{equation*}

we can rewrite (\ref{eq:Vdot}) as an inequality:

\begin{equation}
\begin{aligned}
&\dot{V} -\tau(\|\phi_1\|^2-\alpha_1 ^2\|C_1\chi\|^2)\\
&\leq\chi^T(PA+A^TP+\tau\alpha_1 ^2C_1^TC_1+\frac{1}{\tau}PBB^TP+\frac{1}{\gamma_1} PBB^TP\\
&+\frac{1}{\gamma_2} PBB^TP)\chi+\gamma_1\|\phi_2^*\|^2+ \gamma_2\|\phi_3\|^2.
\end{aligned}
\label{eq:VdotIneq}
\end{equation}

Now suppose that for some $\lambda>0$, the Linear Matrix Inequality
\begin{equation}
\begin{bmatrix}
\begin{aligned}
&PA+A^TP+\tau \alpha_1^2C_1^TC_1 \\
&+\lambda P
\end{aligned} &PB & PB& PB\\
B^TP &-\tau I &0 &0\\
B^TP  &0 &-\gamma_1 I &0\\
B^TP  &0 &0 &-\gamma_2 I\\
\end{bmatrix} <0
\label{eq:LMI}
\end{equation}

admits a solution $P=P^T>0, \: \tau>0, \: \gamma_1>0, \: \gamma_2 >0$. Then we have the following theorem as our main result.

\begin{theorem}
Suppose the function $f$ satisfies conditions (i)-(vi) in \textit{Assumption} \ref{as1}. Also, suppose that there exist $K_0$, $K_1$, $K_2$ and $\lambda>0$ such that the LMI (\ref{eq:LMI}) is feasible; i.e., there exists a matrix $P>0$ and constants $\tau>0, \gamma_1>0, \gamma_2>0$ such that (\ref{eq:LMI}) holds. Then the simultaneous extremum seeking and formation acquisition algorithm converges in the sense that $\lim_{t \to \infty}\|x_1(t)-x_0(t)\| = r_1^*$, $\lim_{t \to \infty}\|x_2(t)-x_0(t)\| = r_2^*$ and
\begin{equation}
\begin{aligned}
\limsup_{t \to \infty}
& \|x_0(t)-x^*\|^2 \\
 &\leq \frac{\gamma_1 k_0^2M_H^2 \|R^{*-1}\|^2(\|r_1^*\|^2+\|r_2^*\|^2)^2}{4 \lambda \sigma_{min}(P)}.
\label{eq:thMain}
\end{aligned}
\end{equation}
\label{th:mainResult}
\end{theorem}

The proof of \textit{Theorem} \ref{th:mainResult} is provided in the Appendix.

\section{Simulation Results}
The proposed extremum seeking control scheme is tested Matlab Simulink environment.

We first set up the simulation using a field function with elliptical level sets of the form
\begin{equation}
f(x,y)=1000\cdot e^{-\frac{(x-100)^2}{70000}-\frac{(y-100)^2}{70000}}
\label{eq:fieldCirc}
\end{equation}

Control gains values used for simulations are $k_1=k_2=0.05$, $k_0=0.7$ and the formation size $\|r^*_{1,2}\|=0.4$ meters.

Search results for the field (\ref{eq:fieldCirc}) are shown in figure \ref{fig:exp1}.
\begin{figure}[h!]
\centering
\includegraphics[scale=0.32]{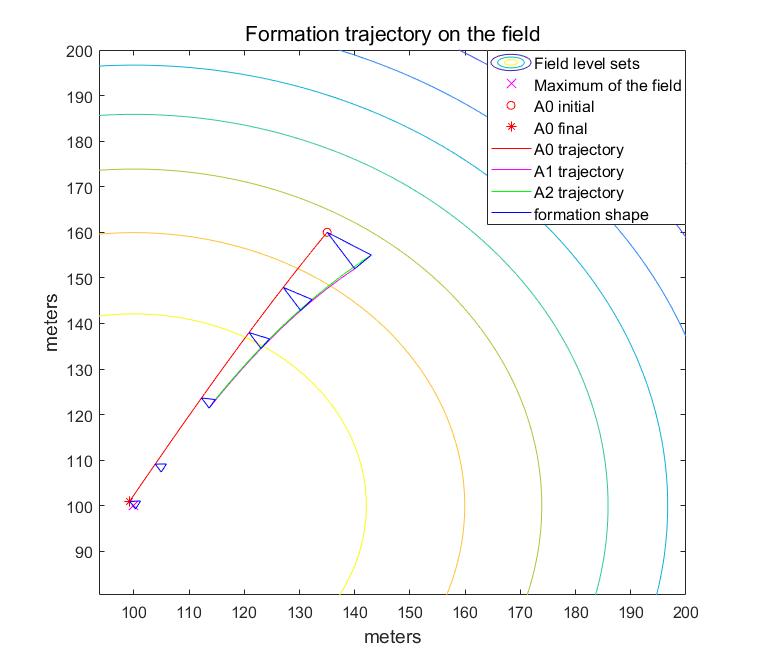}
\includegraphics[scale=0.35]{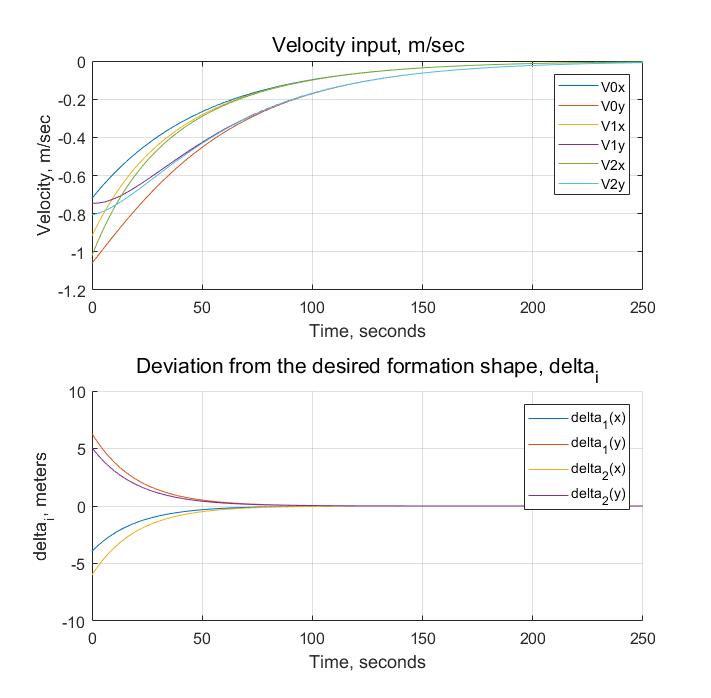}
\caption{\small{Simulation results for the field with elliptical level sets}}\label{fig:exp1}
\end{figure}

For the function (\ref{eq:fieldCirc}) with $\lambda=0.01$, LMI (\ref{eq:LMI}) has a solution $\tilde{\tau}=\frac{\tau}{\gamma_1}=0.1429$, $\tilde{\gamma_2}=\frac{\gamma_2}{\gamma_1}=85901$ and $\sigma_{min}(\tilde{P})=\frac{\sigma_{min}(P)}{\gamma_1}=0.2$.

The numerical value of the theoretical bound in (\ref{eq:thMain}) obtained from the LMI solution is $\lim_{t \to \infty} \|x_0(t)-x^* \|=2.53$ meters. Final extremum seeking control error from simulation is $0.218$ meters.

Simulation experiments have shown that the proposed control algorithm works for fields with more general level sets, that are only locally convex. Simulation results  for the field (\ref{eq:field2}) are presented in  figure \ref{fig:exp2}.

Both simulation experiments have shown that the agents successfully acquire the desired formation shape and converge to the neighbourhood of the maximum of the field.

\begin{equation}
\begin{aligned}
f(x,y)=&e^{- \Big ( \frac{x-100}{100} \Big )^2- \Big ( \frac{y-100}{100} \Big )^2 }\\
&+e^{ -\frac{((x-100)+(y-100))^2}{707}
-\frac{(-(x-100)+(y-100))^2}{143} }\\
&+e^{-\frac{(x-100)^2}{1000}-\frac{(y-100)^2}{50} }
\end{aligned}
\label{eq:field2}
\end{equation}

\begin{figure}[h!]
\centering
\includegraphics[scale=0.35]{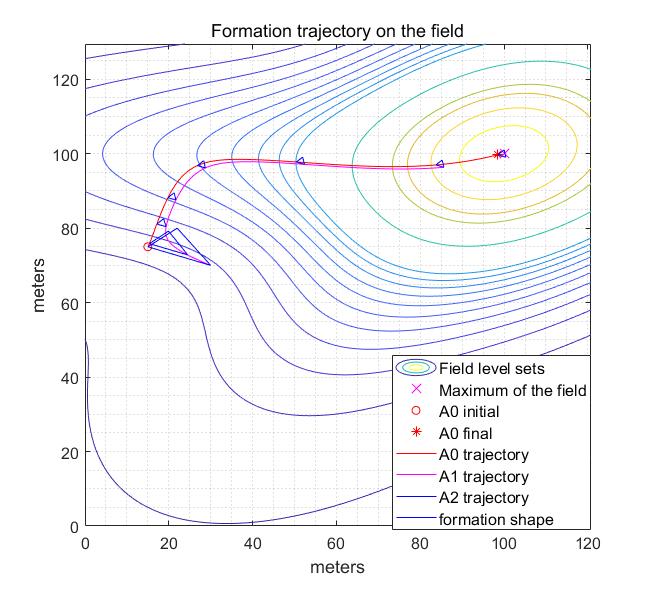}
\includegraphics[scale=0.37]{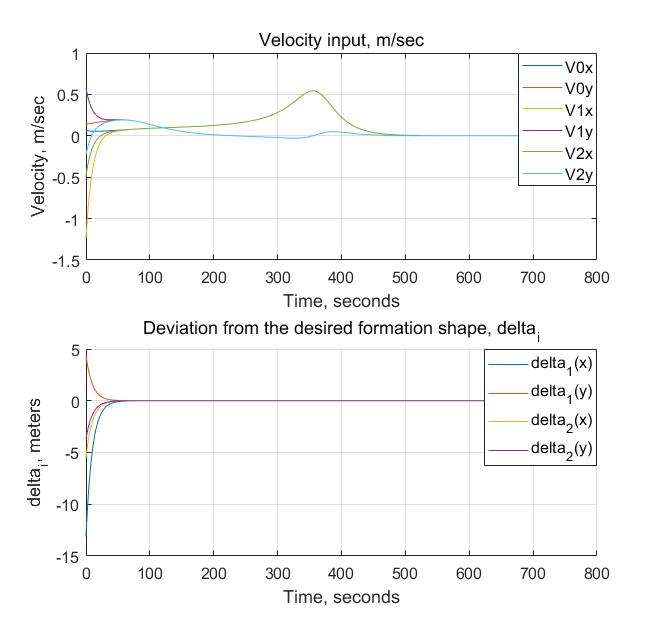}
\caption{\small{Simulation results for a field with  non-circular level sets}}\label{fig:exp2}
\end{figure}

\section{Conclusion}
We have presented a simultaneous formation acquisition and extremum seeking control scheme for a team of three robots to locate a maximum point of a two dimensional signal field. We have shown that the proposed algorithm guarantees convergence to a specified neighbourhood of the maximum of the field while ensuring that the desired formation is acquired and maintained. Our analytical results are supported by simulations, which have demonstrated that the field only has to be locally convex for the algorithm to work.

In future studies, we plan to extend our control method to unicycle robot kinematics and use control techniques that would allow a formation to rotate and change its size in order to improve extremum seeking performance and applicability to more general non-concave settings. We also plan to study the effects of measurement noise and work on a practical implementation of the proposed scheme.

\section*{Appendix 1}
\subsection{Proof of Lemma \ref{lm:phiBounds}}
To find an upper bound on the uncertainty input $\phi_1(t)$, note that the matrix $\hat{\nabla} ^2 f(\xi_{01},\xi_{02})=\begin{bmatrix}
                        \partial_{xx}f(\xi_{01}) & \partial_{xy}f(\xi_{01}) \\
                        \partial_{xy}f(\xi_{02}) & \partial_{xy}f(\xi_{02}) \\
                        \end{bmatrix}$ 
                        satisfies:

\begin{equation}
\begin{aligned}
\hat{\nabla} ^2 f(\xi_{01},\xi_{02})
                        =&\begin{bmatrix} 1&0\\ 0&0 \end{bmatrix}
                        \begin{bmatrix}
                        \partial_{xx}f(\xi_{01}) & \partial_{xy}f(\xi_{01}) \\
                        \partial_{xy}f(\xi_{01}) & \partial_{xy}f(\xi_{01}) \\
                        \end{bmatrix}\\
                        &+\begin{bmatrix} 0&0\\ 0&1 \end{bmatrix}
                        \begin{bmatrix}
                        \partial_{xx}f(\xi_{02}) & \partial_{xy}f(\xi_{02}) \\
                        \partial_{xy}f(\xi_{02}) & \partial_{xy}f(\xi_{02}) \\
                        \end{bmatrix} =\\ &
                        \begin{bmatrix} 1&0\\ 0&0 \end{bmatrix} \nabla ^2 f(\xi_{01})
                        +\begin{bmatrix} 0&0\\ 0&1 \end{bmatrix} \nabla ^2 f(\xi_{02}).
\end{aligned}
\label{eq:AppxnablaXi01Xi02}
\end{equation}

Using (\ref{eq:AppxnablaXi01Xi02}), we rewrite $B_{\beta}$ as:
\begin{equation}
\begin{aligned}
B_{\beta}=&K_0 \begin{bmatrix} 1&0\\ 0&0 \end{bmatrix} (\nabla ^2 f(\xi_{01})-\nabla ^2 f(x^*))\\
&+ K_0 \begin{bmatrix} 0&0\\ 0&1 \end{bmatrix} (\nabla ^2 f(\xi_{02})-\nabla ^2 f(x^*)).
\end{aligned}
\label{eq:AppxBbetaExp}
\end{equation}

Using (\ref{eq:AppxBbetaExp}) and \textit{Assumption} \ref{as1} (vi), the norm of $B_{\beta}$ is:

\begin{equation}
\begin{aligned}
\|B_{\beta}\|\leq &\|K_0\| \Bigg \| \Bigg( \begin{bmatrix} 1&0\\ 0&0 \end{bmatrix} (\nabla ^2 f(\xi_{01})-\nabla ^2 f(x^*)) \\
& +\begin{bmatrix} 0&0\\ 0&1 \end{bmatrix} (\nabla ^2 f(\xi_{02})-\nabla ^2 f(x^*))\Bigg)\Bigg \| \\
\leq & 
\|K_0\|\Bigg( \Bigg\| \begin{bmatrix} 1&0\\ 0&0 \end{bmatrix}\Bigg\| \|\nabla ^2 f(\xi_{01})-\nabla ^2 f(x^*)\| \\ &+ \Bigg \|\begin{bmatrix} 0&0\\ 0&1 \end{bmatrix}\Bigg \| \|\nabla ^2 f(\xi_{02})-\nabla ^2 f(x^*)\|\Bigg)\leq 2G_Hk_0,
\end{aligned}
\end{equation}

Hence, we find a bound on the $\phi_1$ uncertainty input as:
\begin{equation}
\|\phi_1\|\leq \alpha_1 \|C_1 \chi\|,
\label{eq:Appxphi1norm}
\end{equation}

where $\alpha_1=2G_Hk_0$.

The bound on the uncertainty input $\phi_2 ^*$ follows from the definition (\ref{eq:phiDef}), \textit{Assumption} \ref{as1}(iii) and the properties of definite integrals.\\

From the definition of $\phi_2 ^*$ (\ref{eq:phiDef}):
\begin{equation}
\begin{aligned}
\|\phi_2 ^*\|&\leq \frac{1}{2}\|K_0\|\|R^{*-1}\|\\
& \times (\Big |\int _0 ^1 (1-\xi_1)r_1^{*T} \nabla ^2f(x_0+\xi_1 r_1^*)r_1^* d\xi_1\Big |^2\\
&+\Big |\int _0 ^1 (1-\xi_2)r_2^{*T} \nabla ^2f(x_0+\xi_2 r_2^*)r_2^* d\xi_2\Big |^2)^{\frac{1}{2}}.
\end{aligned}
\label{eq:AppPhi2s1}
\end{equation}

We know from the property of define integrals that:
\begin{equation}
\begin{aligned}
&\Big |\int _0 ^1 (1-\xi_i)r_i^{*T} \nabla ^2f(x_0+\xi_1 r_i^*)r_i^* d\xi_i \Big |\\
&\leq \int _0 ^1\Big  \|(1-\xi_i)r_i^{*T} \nabla ^2f(x_0+\xi_1 r_i^*)r_i^* \Big \|d\xi_i\\
&\leq \int _0 ^1|(1-\xi_i)|\|r_i^{*T}\|\|\nabla ^2f(x_0+\xi_1 r_i^*)\|\|r_i^*\|d\xi_i\\
&\leq M_H\|r_i^*\|^2,
\end{aligned}
\label{eq:AppPhi2s2}
\end{equation}

where
\begin{equation*}
\begin{aligned}
|(1-\xi_i)| &\leq 1, \\
\|\nabla ^2f(x_0+\xi_1 r_i^*)\|&\leq M_H.
\end{aligned}
\end{equation*}

From (\ref{eq:AppPhi2s1}) and (\ref{eq:AppPhi2s2}), we obtain the bound on $\phi_2 ^*$:
\begin{equation}
\|\phi_2 ^*\|\leq\frac{k_0M_H}{2}\|R^{*-1}\| (\|r_1^*\|^2+\|r_2^*\|^2).
\label{eq:AppPhi2s}
\end{equation}

Next, we establish a bound on the uncertainty input $\phi_3(t)$. By adding and subtracting\\$\frac{1}{2}K_0R^{*-1}\begin{bmatrix}
\int _0 ^1 (1-\xi_1)r_1^T \nabla ^2f(x_0+\xi_1 r_1)r_1 d\xi_1 \\
\int _0 ^1 (1-\xi_2)r_2^T \nabla ^2f(x_0+\xi_2 r_2)r_2 d\xi_2
\end{bmatrix}$,\\ we rewrite $\phi_3(t)$ as :\\
\begin{equation}
\begin{aligned}
& \phi_3(t)=\phi_3^*(t) + \tilde{\phi_3(t)},
\end{aligned}
\end{equation}
where 

\begin{equation}
\begin{aligned}
& \phi_3^*(t) =\\
&\frac{1}{2}K_0(R^{-1}-R^{*-1})\begin{bmatrix}
\int _0 ^1 (1-\xi_1)r_1^T \nabla ^2f(x_0+\xi_1 r_1)r_1 d\xi_1 \\
\int _0 ^1 (1-\xi_2)r_2^T \nabla ^2f(x_0+\xi_2 r_2)r_2 d\xi_2
\end{bmatrix},\\
&\tilde{\phi}_3(t)=\frac{1}{2}K_0R^{*-1}\Bigg[\begin{matrix}
\int _0 ^1 \Big(r_1^T \nabla ^2f(x_0+\xi_1 r_1)r_1 \\
\int _0 ^1 \Big(r_2^T \nabla ^2f(x_0+\xi_2 r_2)r_2
\end{matrix}\\
&\begin{matrix}
-r_1^{*T} \nabla ^2f(x_0+\xi_1 r_1^*)r_1^*\Big)(1-\xi_1)d\xi_1 \\
-r_2^{*T}  \nabla ^2f(x_0+\xi_2 r_2^*)r_2^*\Big)(1-\xi_2)d\xi_2
\end{matrix}\Bigg].
\end{aligned}
\end{equation}

From \textit{Assumption} \ref{as1}(iii):

\begin{equation}
\begin{aligned}
\|\phi_3^*(t)\|&\leq \frac{k_0M_H}{2}\|R(t)^{-1} - R^{*-1}\|(\|r_1\|^2+\|r_2\|^2).
\end{aligned}
\label{eq:AppxPhi3s}
\end{equation}

We know that $r_i(t)$ converges to a constant value $r_i^*$, therefore $(\|r_1\|^2+\|r_2\|^2)$ is bounded. Also, from (\ref{eq:deltaioft}), $R(t) \to R^*$ exponentially fast and we show in the next subsection that $R^{-1}(t) \to R^{*-1}$ at the same rate.  Hence, $\phi_3^*(t) \to 0$ exponentially fast as $t \to \infty$ .\\

To find a bound on $\tilde{\phi}_3(t)$, we analyse the function under the integral sign:\\
\begin{equation}
\begin{aligned}
&r_i^T \nabla ^2f(x_0+\xi_i r_i)r_i-r_i^{*T}  \nabla ^2f(x_0+\xi_i r_i^*)r_i^*\\
=&r_i^T \nabla ^2f(x_0+\xi_i r_i)r_i-r_i^{T}  \nabla ^2f(x_0+\xi_i r_i^*)r_i\\
&+r_i^{T}  \nabla ^2f(x_0+\xi_i r_i^*)r_i-r_i^{*T}  \nabla ^2f(x_0+\xi_i r_i^*)r_i\\
=&r_i^{T}(\nabla ^2f(x_0+\xi_i r_i)-\nabla ^2f(x_0+\xi_i r_i^*))r_i\\
&+r_i^{T}  \nabla ^2f(x_0+\xi_i r_i^*)r_i-r_i^{*T}  \nabla ^2f(x_0+\xi_i r_i^*)r_i^*.
\end{aligned}
\end{equation}

Using \textit{Assumption} \ref{as1}(v)
\begin{equation}
\begin{aligned}
&\| r_i^{T}(\nabla ^2f(x_0+\xi_i r_i)-\nabla ^2f(x_0+\xi_i r_i^*))r_i \|\\
& \leq \|r_i \| ^2 \| \nabla ^2f(x_0+\xi_i r_i)-\nabla ^2f(x_0+\xi_i r_i^*)\|\\
& \leq \|r_i \| ^2 L_H \| r_i - r_i^*\|,
\end{aligned}
\label{eq:AppxC}
\end{equation}
$r_i(t) \to r_i^*$ exponentially fast and $\|r_i\|^2$ is bounded. Therefore,
$\| r_i^{T}(\nabla ^2f(x_0+\xi_i r_i)-\nabla ^2f(x_0+\xi_i r_i^*))r_i \|\to 0$ exponentially fast as $t \to \infty$ .

Using \textit{Assumption} \ref{as1}(iii)\\
\begin{equation}
\begin{aligned}
& \| r_i^{T}  \nabla ^2f(x_0+\xi_i r_i^*)r_i-r_i^{*T}  \nabla ^2f(x_0+\xi_i r_i^*)r_i^*\|\\
& \leq M_H \|r_i - r_i^* \| \|r_i + r_i^* \|.
\end{aligned}
\end{equation}
Next, we note that $\|r_i + r_i^* \| \leq \|r_i - r_i^* \| + 2\|r_i^*\|$. Thus

\begin{equation}
\begin{aligned}
& \| r_i^{T}  \nabla ^2f(x_0+\xi_i r_i^*)r_i-r_i^{*T}  \nabla ^2f(x_0+\xi_i r_i^*)r_i^*\|\\
& \leq M_H ( \|r_i - r_i^* \|^2 + 2\|r_i^*\|\|r_i - r_i^* \|),
\end{aligned}
\label{eq:AppxB}
\end{equation}
with $r_i(t) \to r_i^*$ exponentially fast and $\|r_i^*\|$ is a constant. Therefore, we conclude that\\ $\| r_i^{T}  \nabla ^2f(x_0+\xi_i r_i^*)r_i-r_i^{*T}  \nabla ^2f(x_0+\xi_i r_i^*)r_i^*\| \to 0$\\ exponentially fast as $t \to \infty$ .

From (\ref{eq:AppxPhi3s}), (\ref{eq:AppxC}) and (\ref{eq:AppxB}), we conclude that\\
\begin{equation}
\| \phi_3 (t)\|=\|\phi_3^*(t) + \tilde{\phi}_3(t)\| \to 0, ~ as~ t\to \infty
 \label{eq:AppPhi3Lim}
\end{equation}
exponentially fast. Hence, we can define a bound on $\phi_3 (t)$ as
\begin{equation}
\|\phi_3(t)\|=\varepsilon e^{- \beta t},
\label{eq:AppPhi3ExpBound}
\end{equation}
where $\varepsilon, ~\beta $ are some positive constants.

\subsection{Proof that $R^{-1}(t)\to R^{*-1}$ exponentially}
By definition $R(t)=\begin{bmatrix}r_1(t)^T\\r_2(t)^T\end{bmatrix}$ and $R^*=\begin{bmatrix}r_1^{*T}\\r_2{*T}\end{bmatrix}$.
In (\ref{eq:deltai})-(\ref{eq:rioft}) we proved, that $r_1(t)\to r_1^*$ and $r_2(t)\to r_2^*$ exponentially. Therefore, for any sufficiently small $\varepsilon>0$, there exists $t_{\varepsilon}$ such that, for all $t>t_{\varepsilon}$:
\begin{equation}
\begin{aligned}
\|r_1(t)\|&>\|r_1^*\|-\varepsilon, \\
\|r_2(t)\|&>\|r_2^*\|-\varepsilon.
\end{aligned}
\label{eq:AppRvareps}
\end{equation} 

Additionally, Corollary \ref{cor:formation} states that:
\begin{equation}
|\det R(t)|> min(|\rho_0|,~ |\rho^*|)\cdot \|r_1(t)\|\cdot \|r_2(t)\|,
\label{eq:AppRCor2}
\end{equation}
which together with (\ref{eq:AppRvareps}) implies the following inequality
\begin{equation}
\begin{aligned}
|\det R(t)|&> min(|\rho_0|,~ |\rho^*|)\cdot(\|r_1^*\|-\varepsilon )\cdot (\|r_2^*\|-\varepsilon )\\
&\overset{\Delta}{=} \bar{\varepsilon} >0,~\forall t>t_{\varepsilon}.
\label{eq:AppRCor2Upd}
\end{aligned}
\end{equation}
Furthermore, using notation in (\ref{eq:ri}), we can write
\begin{equation}
\begin{aligned}
R^{-1}(t)=\frac{1}{\det R(t)}\begin{bmatrix}r_{2y}(t)&-r_{1y}(t)\\
-r_{2x}(t)&r_{1x}(t)
\end{bmatrix}.
\label{eq:AppRRinv}
\end{aligned}
\end{equation}
Therefore, from (\ref{eq:AppRCor2Upd}) and (\ref{eq:AppRRinv}), it follows that for all $t>t_{\varepsilon}$
\begin{equation}
\begin{aligned}
\|R^{-1}(t)\|<\frac{1}{\bar{\varepsilon}} \left \| \begin{bmatrix}r_{2y}(t)&-r_{1y}(t)\\
-r_{2x}(t)&r_{1x}(t)
\end{bmatrix}\right \|.
\label{eq:AppRRinvNorm}
\end{aligned}
\end{equation}
Since $r_i(t)\to r_i^*$, we can ascertain that $r_{1x}(t)$, $r_{1y}(t)$, $r_{2x}(t)$ and $r_{2y}(t)$ are bounded, which implies that there exists $\hat{\varepsilon}>0$ such that
\begin{equation}
\begin{aligned}
\|R^{-1}(t)\|<\frac{1}{\bar{\varepsilon}}\hat{\varepsilon} ,~\forall t>t_{\varepsilon}.
\label{eq:AppRRinvUpd}
\end{aligned}
\end{equation}
Using the following identity 
\begin{equation}
\|R^{-1}(t)-R^{*-1}\|=\|R^{-1}(t)(R(t)-R^*)R^{*-1}\|,
\end{equation}
we further obtain
\begin{equation}
\begin{aligned}
\|R^{-1}(t)-R^{*-1}\|&\leq\|R^{-1}(t)\|\cdot \|R(t)-R^*\|\cdot \|R^{*-1}\| \\
&< \frac{\hat{\varepsilon}}{\bar{\varepsilon}}\|R^{*-1}\| \cdot \|R(t)-R^*\|,~\forall t>t_{\varepsilon}.
\label{eq:AppRRinvDiff}
\end{aligned}
\end{equation}
Thus, $R^{-1}(t)\to R^{*-1}$ as required.

\subsection{Proof of Theorem \ref{th:mainResult}}

If LMI (\ref{eq:LMI}) admits a solution $P=P^T>0$, $\tau>0$, $\gamma_1>0$, $\gamma_2>0$, then for all $\phi_1,~ \phi_2^*, ~\phi_3$ and all $t>t_{\varepsilon}$
\begin{equation}
\begin{aligned}
&\dot{V} + \lambda V <\tau(\|\phi_1\|^2-\alpha_1^2\|C_1\chi\|^2)\\
&+\gamma_1\|\phi_2^*\|^2+ \gamma_2\|\phi_3\|^2,
\end{aligned}
\end{equation}
and taking into account (\ref{eq:constraints1}), (\ref{eq:constraints2}), we obtain that for all admissible $\phi_1,~ \phi_2^*,~ \phi_3$ and $\forall t>t_{\varepsilon}$,
\begin{equation}
\dot{V}+ \lambda V < \gamma_1 \alpha_2^{*2} + \gamma_2\varepsilon^2e^{-2\beta t}.
\label{eq:AppVdotineq}
\end{equation}

According to the Gronwall-Bellman Lemma \cite{khalil1996noninear}, inequality (\ref{eq:AppVdotineq}) implies that:\\
\begin{equation}
\begin{aligned}
V(t)\leq& e^{-\lambda t}V(0)+ \gamma_1 \alpha_2^{*2} \int _{0}^t e^{-\lambda(t-\tau)}d\tau\\
&+\gamma_2\varepsilon^2\int _{0}^t e^{-2\beta \tau}e^{-\lambda(t-\tau)}d\tau\\
 \leq& e^{-\lambda t}V(0)+\gamma_1 \alpha_2^{*2} \frac{(1-e^{-\lambda t})}{\lambda}\\
 &+\gamma_2 \varepsilon^2\frac{(e^{-2\beta t}-e^{-\lambda t})}{2\beta -\lambda}.
\end{aligned}
\label{eq:AppVineq}
\end{equation}

This implies that for any initial condition $\chi(0)$ the trajectories of the system (\ref{eq:ChiDot}), (\ref{eq:constraints1}), (\ref{eq:constraints2}) satisfy
\begin{equation}
\begin{aligned}
\lim_{t \to \infty} (\| \delta_1(t)\|^2+\| \delta_2(t)\|^2+\| z(t)\|^2) \leq \frac{\gamma_1\alpha_2^{*2}}{\lambda \sigma_{min}(P)}
\end{aligned}
\label{eq:AppLimChi}
\end{equation}

Note that from (\ref{eq:deltaioft}), $\delta_1(t) \to 0$, $\delta_2(t) \to 0$ exponentially fast. Therefore, we can rewrite (\ref{eq:AppLimChi}) as

\begin{equation}
\begin{aligned}
\lim_{t \to \infty} (\| z(t)\|^2) & \leq \frac{\gamma_1\alpha_2^{*2}}{\lambda \sigma_{min}(P)}\\
 &\leq \frac{\gamma_1k_0^2M_H^2\|R^{*-1}\|^2(\|r_1^*\|^2+\|r_2^*\|^2)^2}{\lambda \sigma_{min}(P)}.
\end{aligned}
\label{eq:AppLimZ}
\end{equation}

\bibliographystyle{IEEEtrans}
\bibliography{AdaptiveLocalization,VectorCalcOpt,mainESC}

\end{document}